\documentclass{article}

\usepackage{arxiv}

\usepackage[utf8]{inputenc} 
\usepackage[T1]{fontenc}    
\usepackage{hyperref}       
\usepackage{url}            
\usepackage{booktabs}       
\usepackage{amsfonts}       
\usepackage{nicefrac}       
\usepackage{microtype}      
\usepackage{lipsum}		
\usepackage{graphicx}

\usepackage[brazil,english]{babel}
\usepackage[utf8]{inputenc}
\usepackage{ae}

\usepackage{subfigure}
\usepackage{amsmath} 
\usepackage{amsthm}
\usepackage{amssymb}  
\usepackage{bm}

\newtheorem{theorem}{Theorem}
\newtheorem{corollary}{Corollary}[theorem]
\newtheorem*{remark}{Remark}

\title{Controlling epidemic diseases based only on social distancing level}

\author{\includegraphics[scale=0.06]{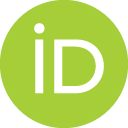}\hspace{1mm}Samaherni M. Dias \\
	\texttt{sama@laci.ufrn.br} \\
	\And
	\includegraphics[scale=0.06]{orcid.png}\hspace{1mm}Kurios I. P. de M. Queiroz \\
	\texttt{kurios@laci.ufrn.br} \\	
	\And
	\includegraphics[scale=0.06]{orcid.png}\hspace{1mm}Allan de M. Martins \\
	\texttt{allan@laci.ufrn.br} \\
	\And
	Laboratory of Automation, Control, and Instrumentation (LACI)\\
	Department of Electrical Engineering\\
	Federal University of Rio Grande do Norte (UFRN)\\
	Natal-RN, Brasil \\
}

\date{}

\begin{document}
\maketitle

\begin{abstract}
The World Health Organization (WHO) made the assessment that COVID-19 can be characterized as a pandemic on March 11, 2020. To the COVID-19 outbreak, there is no vaccination and no treatment. The only way to control the COVID-19 outbreak is sustained physical distancing. In this work, a simple control law was proposed to keep the infected individuals during the COVID-19 outbreak below the desired number. The proposed control law keeps the value of infected individuals controlled only adjusting the social distancing level. The stability analysis of the proposed control law is done and the uncertainties in the parameters were considered. A version of the proposed controller to daily update was developed. This is a very simple approach to the developed control law and can be calculated in a spreadsheet. In the end, numerical simulations were done to show the behavior of the number of infected individuals during an epidemic disease when the proposed control law is used to adjust the social distancing level.
\end{abstract}


\section{Introduction}

As of March 11, 2020, the World Health Organization (WHO) made the assessment that COVID-19 can be characterized as a pandemic. This assessment changed the attention of the researchers to the novel severe acute respiratory syndrome coronavirus 2 (SARS-CoV-2). 

In a realistic scenario of an epidemic disease, the available social and medical resources to treat diseases or to prevent their spreading are usually limited \cite{Jiang_2018}. The COVID-19 pandemic has the basic reproduction number relatively high and presents worrying hospitalization and death rates \cite{Verity_2020}. 

Measures for prevention and control of infectious diseases include vaccination, treatment, quarantine, isolation, and prophylaxis. Quarantine and isolation are two measures by which exposed or infectious individuals are removed from the population to prevent further spread of the infection. Quarantine is applied less often. It is one of the first response methods that can be used in an extreme emergency. Quarantine was implemented during the SARS epidemic of 2002–2003\cite{martcheva_2015}. 

In the literature there are several works considering optimal control applied to epidemic diseases \cite{martcheva_2015,Ball_2008,Gaff_2009,Zaman_2009}. Optimal control theory provides a valuable tool to begin to assess the trade-offs between vaccination and treatment strategies \cite{Gaff_2009}. However, for the recent COVID-19 outbreak there is no vaccination and treatment. The only possible to control this outbreak is sustained physical distancing.

In the work of Kiesha Prem et al. \cite{Prem_2020}, the authors conclude that non-pharmaceutical interventions based on sustained physical distancing have a strong potential to reduce the magnitude of the epidemic peak of COVID-19 and lead to a smaller number of overall cases. In this work, the authors used an age-group model because the social mixing patterns vary across locations, including households, workplaces, schools, and other locations. 

In the work of Joel Hellewell et al. \cite{Hellewell_2020}, the authors conclude that in most scenarios, highly effective contact tracing and case isolation is enough to control a new outbreak of COVID-19.

Cameron Nowzari Et al. \cite{Nowzari_2016} had a survey about the analysis and control of epidemics. In that work, the authors highlight as main research challenges: ``all control methods discussed so far have been for deterministic models"; ``all control methods discussed so far have admitted centralized solutions"; ``all control methods discussed so far have assumed there are no uncertainties"; ``more general epidemic models are needed". 

Based on the challenges shown and the characteristics of the COVID-19 outbreak, the focus of this work is to propose a control law to control outbreaks of the epidemic disease there is no vaccination and treatment. The control law proposed here is based on the adjustment of the social distancing level. In addition, a more general epidemic model will be presented.      

The paper is structured as follows. In Section 2, the epidemic models are introduced. In Section 3, a controller law, considering the uncertainties of the process, will be presented to a classical model of the epidemic disease. Section 4 contains some numerical simulations of the theory, and Section 5 provides some concluding comments. 

\section{SIR model}

The dynamics of epidemic diseases can be described by Susceptible-Infectious-Recovered (SIR) model. The SIR model was formulated por Kermack and McKendrick in 1927\cite{Kermack_1927}. The SIR model applies to epidemics having a relatively short duration (few months). The models presented in this work is to epidemic diseases with no treatment and vaccination.

The SIR model is a system of three differential equations:
\begin{equation}
\begin{array}{rcl}
\dfrac{d~s(t)}{dt}&=&-\beta(t) i(t)s(t)  \\[1em]
\dfrac{d~i(t)}{dt}&=&\beta(t) i(t)s(t) - \gamma i(t)\\[1em]
\dfrac{d~r(t)}{dt}&=&\gamma i(t)
\end{array}\label{eq:SIR}
\end{equation}where $s(t)$ is the number of susceptible individuals, $i(t)$ is the number of infected individuals, $r(t)$ is the number of recovered, $\beta(t)$ is the proportional coefficient of the disease transmission rate and $\gamma$ is the recovery rate (the specific rate at which infected individuals recover from the disease). Let $N(t)=s(t)+i(t)+r(t)$ the number of individuals in the population and considering the population constant 
\begin{equation}
\begin{array}{rcl}
\dfrac{d~N(t)}{dt}&=&\dfrac{d~s(t)}{dt}+\dfrac{d~i(t)}{dt}+\dfrac{d~r(t)}{dt}=0
\end{array}\label{eq:N}
\end{equation}

The proportional coefficient of the disease transmission rate is obtained by 
\begin{equation}
\begin{array}{rcl}
\beta(t)=\dfrac{\kappa(t)\tau(t)}{N(t)}
\end{array}\label{eq:beta}
\end{equation}where $\kappa(t)$ is proportional to the numbers of contacts that an infected individual has per unit time and $\tau(t)$ is the probability that a contact with a susceptible individual results in transmission \cite{martcheva_2015}.

The the basic reproductive number is given by 
\begin{equation}
\begin{array}{rcl}
R_0=\dfrac{\beta_0}{\gamma}=\dfrac{\kappa_0\tau_0}{N\gamma}
\end{array}\label{eq:R0}
\end{equation}where $\kappa_{0}$ represents the mean value of the numbers of contacts that an infected individual has per unit time in normal conditions, $\tau_0$ represents the mean value of the probability that a contact with a susceptible individual results in transmission and $\beta_0$ is the proportional coefficient of the disease transmission rate when calculated by $\kappa_{0}$ and the number of individuals in the population ($N$) is considered constant. 

The fundamental characteristics of the model (\ref{eq:SIR}) are:
\begin{enumerate}
\item[\textbf{C1.}] the number of susceptible individuals is greater than zero ($s(t)>0$);
\item[\textbf{C2.}] the number of infected individuals is greater than zero ($i(t)>0$);
\item[\textbf{C3.}] the value of the numbers of contacts that an infected individual has per unit time is greater than zero ($\kappa(t)>0$);
\item[\textbf{C4.}] the probability that a contact with a susceptible individual results in transmission is greater than zero ($\tau(t)>0$);
\item[\textbf{C5.}] the recovery rate is greater than zero ($\gamma>0$);
\end{enumerate}

\section{Proposed Controller}

The main aim of the proposed controller is to keep the number of hospitalized infected individuals bellow of the capacity of the health care system in a specific region. Social distancing will be used to control the number of hospitalized individuals. The controller proposed here will calculate the exact social distancing rate to keep the number of hospitalized individuals in the desired value. The proposed controller can be applied to help in the reopening of a region.

First it is necessary consider 
\begin{equation}
\kappa(t)=\rho(t)\kappa_0\label{eq:kappa}
\end{equation}where $\rho(t)$ is the control signal and it will modulate the mean value of the numbers of contacts that an infected individual has per unit time in normal conditions ($\kappa_{0}$).  

The purpose is to find a control law $\rho(t)$, to the system (\ref{eq:SIR}), such that the output error
\begin{equation}
e(t)=i_d-i(t),\label{eq:erro}
\end{equation}where $i_d$ indicates the maximum desired infected individuals, tends to zero when $t\rightarrow\infty$, which guarantees that the health care system is not going to collapse.

The following assumptions are made:
\begin{enumerate}
\item[\textbf{A1.}] the maximum desired infected individuals is assumed to be a step and greater than zero ($i_d>0$);
\item[\textbf{A2.}] the number of individuals (N(t)) in the population is assumed to be constant ($N$);
\item[\textbf{A3.}] the probability that contact with a susceptible individual will result in the transmission is assumed to be constant ($\tau_0$);
\end{enumerate}

Considering the assumptions (A2-A3), the proportional coefficient of the disease transmission rate is updated to
\begin{equation}
\beta(t)=\dfrac{\rho(t)\kappa_0\tau_0}{N}=\rho(t)\beta_0\label{eq:beta_t_rho}
\end{equation}

Let us consider the control law 
\begin{equation}
\begin{array}{rcl}
\rho(t)=\dfrac{\psi(t) N}{ i(t)s(t)\kappa_{0}\tau_0}
\end{array}\label{eq:rho}
\end{equation}where
\begin{equation}
\psi(t)=\psi_1 e(t) + \psi_2 \displaystyle\int_{0}^{t}e(t) dt\label{eq:psi}
\end{equation}and $\psi_1$, $\psi_2$ are nonnegative constants chosen to adjust the error dynamics. 

\begin{theorem}
Consider the system of the equation (\ref{eq:SIR}), the error equation (\ref{eq:erro}) and the control law (\ref{eq:rho}). Whenever all assumptions (A1)-(A3) are satisfied, the error $e(t)$ will converge to zero when time tends to infinity.
\end{theorem}

\begin{proof}
Consider the Lyapunov function
\begin{equation}
V(\dot{e},e)=\dfrac{\dot{e}^{2}}{2}+\dfrac{\psi_2 e^{2}}{2}\label{eq:V}
\end{equation}

Then, the time derivative of $V(\dot{e},e)$ will be
\begin{equation}
\dot{V}(\dot{e},e)=\ddot{e}\dot{e}+\psi_2\dot{e}e\label{eq:dV}
\end{equation}

Applying (\ref{eq:SIR}), (\ref{eq:erro}) and (\ref{eq:rho}) in (\ref{eq:dV}), the result is
\begin{equation}
\dot{V}(\dot{e},e)=-\psi_1\dot{e}^{2}-\gamma\dot{e}^{2}\leq 0. \label{eq:dV_final}
\end{equation}

For the system to maintain $\dot{V}(\dot{e},e)=0$ condition, the trajectory must be confined to the line $\dot{e}=0$. Using the system dynamics (\ref{eq:erro}) yields: 
\begin{equation}
\dot{e}\equiv 0\Rightarrow \ddot{e}\equiv 0 \Rightarrow - \psi_2 e \equiv 0 \Rightarrow  e \equiv 0\nonumber
\end{equation}which by LaSalle's theorem the origin is globally asymptotically stable (GAS).  
\end{proof}

Now, consider $\hat{\kappa}_{0}\hat{\tau}_0$ and $\hat{\gamma}$ the known values to the parameters $\kappa_{0}\tau_0$ and $\gamma$, respectively, where $\hat{\kappa}_{0}\hat{\tau}_0=\Delta_{\kappa}\kappa_{0}\tau_0$ and $\hat{\gamma}=\Delta_{\gamma}\gamma$. The  $\Delta_{\kappa}$ and $\Delta_{\gamma}$ represent the uncertainties in the system and it can assume any value between 0.5 and 1.5. The control law will update to 
\begin{equation}
\begin{array}{rcl}
\rho(t)=\dfrac{\psi(t) N}{ i(t)s(t)\hat{\kappa}_{0}\hat{\tau}_0}
\end{array}\label{eq:rho_c1}
\end{equation}

\begin{corollary}
Consider the system of the equation (\ref{eq:SIR}), the error equation (\ref{eq:erro}), the control law (\ref{eq:rho_c1}) and that there are uncertainties in some of the parameters $\kappa_{0}\tau_0$, $\gamma$. Whenever all assumptions (A1)-(A3) are satisfied, the error $e(t)$ will converge to zero when time tends to infinity.
\end{corollary}

\begin{proof}
Consider the Lyapunov function
\begin{equation}
V(\dot{e},e)=\dfrac{\dot{e}^{2}}{2}+\dfrac{\psi_2 e^{2}}{2\Delta_{\kappa}}
\end{equation}

Then, the time derivative of $V(\dot{e},e)$ will be
\begin{equation}
\dot{V}(\dot{e},e)=-\dfrac{\psi_1\dot{e}^{2}}{\Delta_{\kappa}}-\dfrac{\gamma\dot{e}^{2}}{\Delta_{\kappa}}\leq 0.
\end{equation}

For the system to maintain $\dot{V}(\dot{e},e)=0$ condition, the trajectory must be confined to the line $\dot{e}=0$. Using the system dynamics (\ref{eq:erro}) yields: 
\begin{equation}
\dot{e}\equiv 0\Rightarrow \ddot{e}\equiv 0 \Rightarrow - \dfrac{\psi_2}{\Delta_{\kappa}}e  \equiv 0 \Rightarrow  e \equiv 0\nonumber
\end{equation}which by LaSalle's theorem the origin is globally asymptotically stable (GAS).
\end{proof}

To the real epidemic disease without treatment, the control law will update to
\begin{equation}
\rho(t)=\max\left(0,\min\left(1,\dfrac{\psi(t) N}{ i(t)s(t)\kappa_{0}\tau_0}\right)\right).\label{eq:rho_c2}
\end{equation}

In other words, $\rho(t)$ can take any value between 0 and 1. $\rho(t)<0$ occurs when $i(t)>>i_d$ and indicates that is necessary to increase the recovery rate, which is possible only with treatment. $\rho(t)>1$ occurs when $i(t)<<i_d$ and indicates that is necessary to increase the numbers of contacts that an infected individual has per unit time, however, this will not be done because the aim of the control law, when applied to the real epidemic disease, is guarantee $e(t)>0$ ($i(t)<i_d$). 

\begin{corollary}
Consider the system of the equation (\ref{eq:SIR}), the error equation (\ref{eq:erro}) and the control law (\ref{eq:rho_c2}). Whenever all assumptions (A1)-(A3) are satisfied, the error $e(t)$ will converge to a value greater than zero when time tends to infinity.
\end{corollary}

\begin{proof}
Consider the system (\ref{eq:SIR}) when the control law (\ref{eq:rho_c2}) is applied 
\begin{equation}
\begin{array}{rcl}
\dfrac{d~s(t)}{dt}&=&-\rho(t)\beta_{0} i(t)s(t)  \\[1em]
\dfrac{d~i(t)}{dt}&=&\rho(t)\beta_{0} i(t)s(t) - \gamma i(t)\\[1em]
\end{array}\nonumber
\end{equation}

Based on the characteristics (C1-C5) of the system (\ref{eq:SIR}), consider the Lyapunov function
\begin{equation}
V(s,i)=s + i\nonumber
\end{equation}

Then, the time derivative of $V(s,i)$ will be
\begin{equation}
\dot{V}(s,i)=- \gamma i\nonumber
\end{equation}

Therefore $\dot{V}(s,i)\leq0$ which implies that $i(t)\rightarrow 0$ and $e(t) \rightarrow i_{d}$ when $t \rightarrow \infty$.
\end{proof}

\begin{remark}
During the period of control law saturation, the value of the integral term of the equation (\ref{eq:psi}) is not computed.  
\end{remark}

\subsection{Proposed controller using estimates of the $i(t)$}

To real epidemic disease, the number of confirmed cases generally is the most reliable measurement. Based on this, the control law (\ref{eq:rho}) will be adjusted to be applied in the real case. The control law will change to 
\begin{equation}
\rho(t)=\dfrac{\psi(t) N}{\hat{i}(t)s(t)\kappa_{0}\tau}=\dfrac{\psi(t)}{\hat{i}(t)s(t)\beta_{0}}\label{eq:rho_real}
\end{equation}where $\hat{i}(t)$ is the estimated value to $i(t)$. 

To apply the proposed control law will be necessary to know $N$, $\beta_{0}$, to measure $s(t)$ and to estimate $i(t)$. Generally the values of $N$ and $\beta_{0}$ are known ($\beta_{0}$ it is obtained from $R_0$). The value $s(t)$ is measured based on the number of confirmed cases $c(t)$ and $\hat{i}(t)$ is based on $s(t)$ and $\beta_{0}$.

To epidemic disease, based on standard SIR model (\ref{eq:SIR}), where all individual of the population is susceptible, the dynamic of the number of confirmed cases is given by 
\begin{equation}
\dfrac{d~c(t)}{dt}=\beta(t) i(t)s(t)\label{eq:c_real}
\end{equation}

Based on the equation (\ref{eq:c_real}), it is possible to consider
\begin{equation}
s(t)=N-c(t)\label{eq:s_real}
\end{equation}

The estimate of the $i(t)$ is given by 
\begin{equation}
\hat{i}(t)=\dfrac{d~s(t)}{dt}\cdot(\beta(t)s(t))^{-1}=\dfrac{d~s(t)}{dt}\cdot(\rho(t)\beta_0 s(t))^{-1}\label{eq:hi_real}
\end{equation}

\subsubsection{Proposed controller with daily update}

It is important to notice that is easy to apply the proposed controller to real epidemic disease with daily update. First, consider $d$ as the day of the epidemic disease. To calculate 
\begin{equation}
\rho(d)=\dfrac{\psi_1 e(d)+\psi_2 \displaystyle\sum_{k=1}^{d}e(k)}{\hat{i}(d)s(d)\beta_{0}}\label{eq:rho_d}
\end{equation}it is necessary to know the basic reproductive number ($R_0$) to obtain $\beta_0$, it is necessary to obtain 
\begin{equation}
s(d)=N-c(d)\label{eq:sd_real}
\end{equation}where $c(d)$ is the number of confirmed cases at day $d$, it is necessary to calculate 
\begin{equation}
\hat{i}(d)=\dfrac{s(d)-s(d-1)}{\rho(d-1)\beta_0 s(d-1)},\label{eq:hi_d_real}
\end{equation}
\begin{equation}
e(d)=i_d(d)-\hat{i}(d),\label{eq:erro_d}
\end{equation}to define the gains $\psi_1$, $\psi_2$ and to define the maximum desired infected individuals ($i_d$). In the case where $i(d)$ is measured, the value of $\hat{i}(d)$, in the equations (\ref{eq:rho_d}) and (\ref{eq:erro_d}), has to be changed by $i(d)$. 

The implementation of the equation (\ref{eq:rho_d}) in the real outbreak is simple and can be done in a spreadsheet. Consider that the basic reproductive number $R_0$, the recovery rate $\gamma$, the number of population ($N$), the number of susceptible individuals ($s$), and the number of infected individuals ($i$) are known, the isolation level can be computed in a spreadsheet as follows (see Table (\ref{tb:001})).
\begin{table}[!ht]
	\caption{Example of how to use the proposed control law when s and i are measured}
	\centering
	\begin{tabular}{cccccc}
		\toprule
		Day   & Susceptible & Infected & Error  & Accumulated error & Isolation \\
		$d$   & $s(d)$ & $i(d)$ & $e(d)$ & $\sum e$ & 1-$\rho(d)$\\
		\midrule
0 & $s(0)$ & $i(0)$ & $i_d-i(0)$ & $e(0)$ & $1-\dfrac{\psi_1 e(0)+\psi_2 \sum e}{i(0)s(0)\beta_0}$ \\[1em]
1 & $s(1)$ & $i(1)$ & $i_d-i(1)$ & $e(0)+e(1)$ & $1-\dfrac{\psi_1 e(1)+\psi_2 \sum e}{i(1)s(1)\beta_0}$ \\[1em]
2 & $s(2)$ & $i(2)$ & $i_d-i(2)$ & $e(0)+e(1)+e(2)$ & $1-\dfrac{\psi_1 e(2)+\psi_2 \sum e}{i(2)s(2)\beta_0}$ \\
		\bottomrule
	\end{tabular}
	\label{tb:001}
\end{table}

Consider that the basic reproductive number $R_0$, the recovery rate $\gamma$, the number of population ($N$), and the number of confirmed infected individuals ($c$) are known, the isolation level can be computed in a spreadsheet as follows (see Table (\ref{tb:002})).
\begin{table}[!ht]
	\caption{Example of how to use the proposed control law when c is measured}
	\centering
	\begin{tabular}{ccccccc}
		\toprule
		Day   & confirmed & Susceptible & Infected & Error  & Accumulated error & Isolation \\
		$d$   & $c(d)$ & $s(d)$ & $\hat{i}(d)$ & $e(d)$ & $\sum e$ & 1-$\rho(d)$\\
		\midrule
0 & $c(0)$ & $N-c(0)$ & $\hat{i}(0)$ & $i_d-\hat{i}(0)$ & $e(0)$ & $1-\dfrac{\psi_1 e(0)+\psi_2 \sum e}{\hat{i}(0)s(0)\beta_0}$ \\[1em]
1 & $c(1)$ & $N-c(1)$ & $\dfrac{s(1)-s(0)}{\rho(0)\beta_{0}s(0)}$ & $i_d-\hat{i}(1)$ & $e(0)+e(1)$ & $1-\dfrac{\psi_1 e(1)+\psi_2 \sum e}{\hat{i}(1)s(1)\beta_0}$ \\[1em]
2 & $c(2)$ & $N-c(2)$ & $\dfrac{s(2)-s(1)}{\rho(1)\beta_{0}s(1)}$ & $i_d-\hat{i}(2)$ & $e(0)+e(1)+e(2)$ & $1-\dfrac{\psi_1 e(2)+\psi_2 \sum e}{\hat{i}(2)s(2)\beta_0}$ \\
		\bottomrule
	\end{tabular}
	\label{tb:002}
\end{table}

\section{Numerical simulations}

In this section, the proposed controller will be simulated and their results are presented and analyzed. The Figures (1-3) presents the result to the proposed controller applied to the standard SIR model. Another important consideration is that the numerical method for solving ordinary differential equations was the Euler method and the integration step was one day and that all simulations has 600 days.  

\subsection{Standard SIR model}

To the simulations of proposed controller (equation \ref{eq:rho}) applied to standard SIR model (equation \ref{eq:SIR}) it is necessary to define the number of individuals in the population $N$= 1 million, the basic reproductive number $R_0=2$ \cite{NEJMoa_2020,Prem_2020,Kucharski_2020}, the recovery rate $\gamma=0.2$ \cite{NEJMoa_2020,Prem_2020,Kucharski_2020}, the initial condition of the susceptible individuals $s(0)=N-1$, the initial condition of the infected individuals $i(0)=1$, the initial condition of the recovered individual $r(0)=0$, the number of the hospitalized individuals ($H$) is 10\% of the infected individuals, the desired number of hospitalized individuals ($H_d$) is 800 (this value is according with the capacity of the health care of the region) and the gains of the controller are $\psi_1=0.02$, $\psi_2=0.0043$.  

The simulation of Figure \ref{fg:501} shows the behavior of the infected individuals that will need health care when none controller was applied. In this simulation, it is possible to see that an epidemic disease without a controller can result in a huge number of hospitalized individuals and this occurs in a few days. Generally, the number of hospitalized individuals is greater than the capacity of the health care system of the region. To guarantee that the number of hospitalized individuals keep lower than the health care system of the region, the proposed controller will be applied.  
\begin{figure}[ht]
\center
\subfigure[$\{H,H\_d\}(\text{individuals}) \times t(\text{days})$]{\includegraphics[width=0.48\linewidth, height=!]{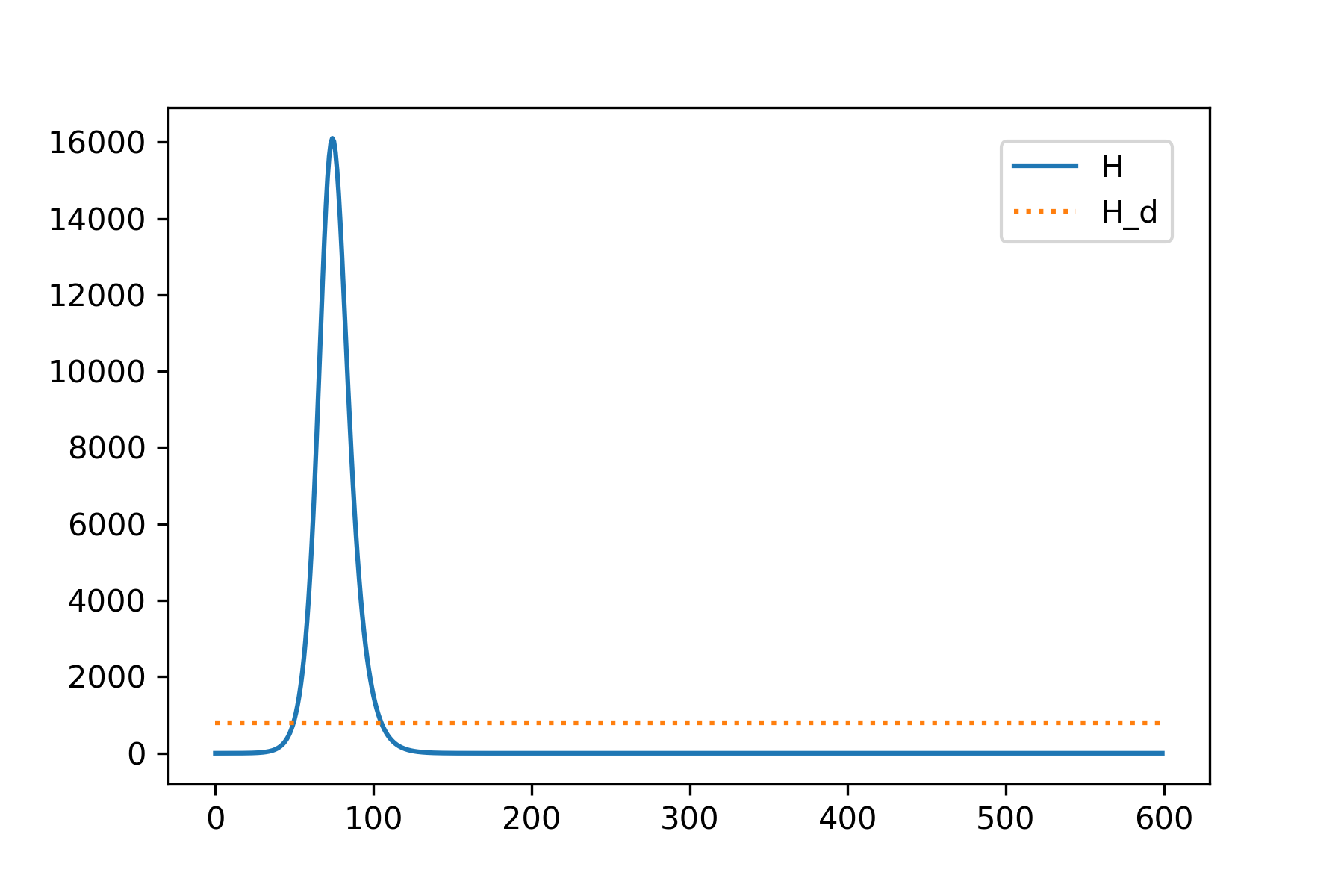}}
\caption{Simulation of standard SIR model without a controller, where blue line is the number of the hospitalized individuals $H$, orange dotted line is the desired number of hospitalized individuals $H_d$}\label{fg:501}
\end{figure}

The simulation of Figure \ref{fg:502} shows the behavior of the infected individuals that will need health care when the proposed controller is applied to adjust the social distancing level. In this simulation, it is possible to see that the number of hospitalized individuals keeps lower than the desired number of hospitalized individuals. However to do this, the number of days with any level of the social distancing is great and inversely proportional to the desired number of hospitalized individuals. In the beginning, the social distancing is near 50\% and shows a soft decay during 300 days. The uncertainties considered was +15\% and -20\% in the values of $\beta_{0}$ and $\gamma$, respectively.   
\begin{figure}[ht]
\center
\subfigure[$\{H,H\_u,H\_d\}(\text{individuals}) \times t(\text{days})$]{\includegraphics[width=0.48\linewidth, height=!]{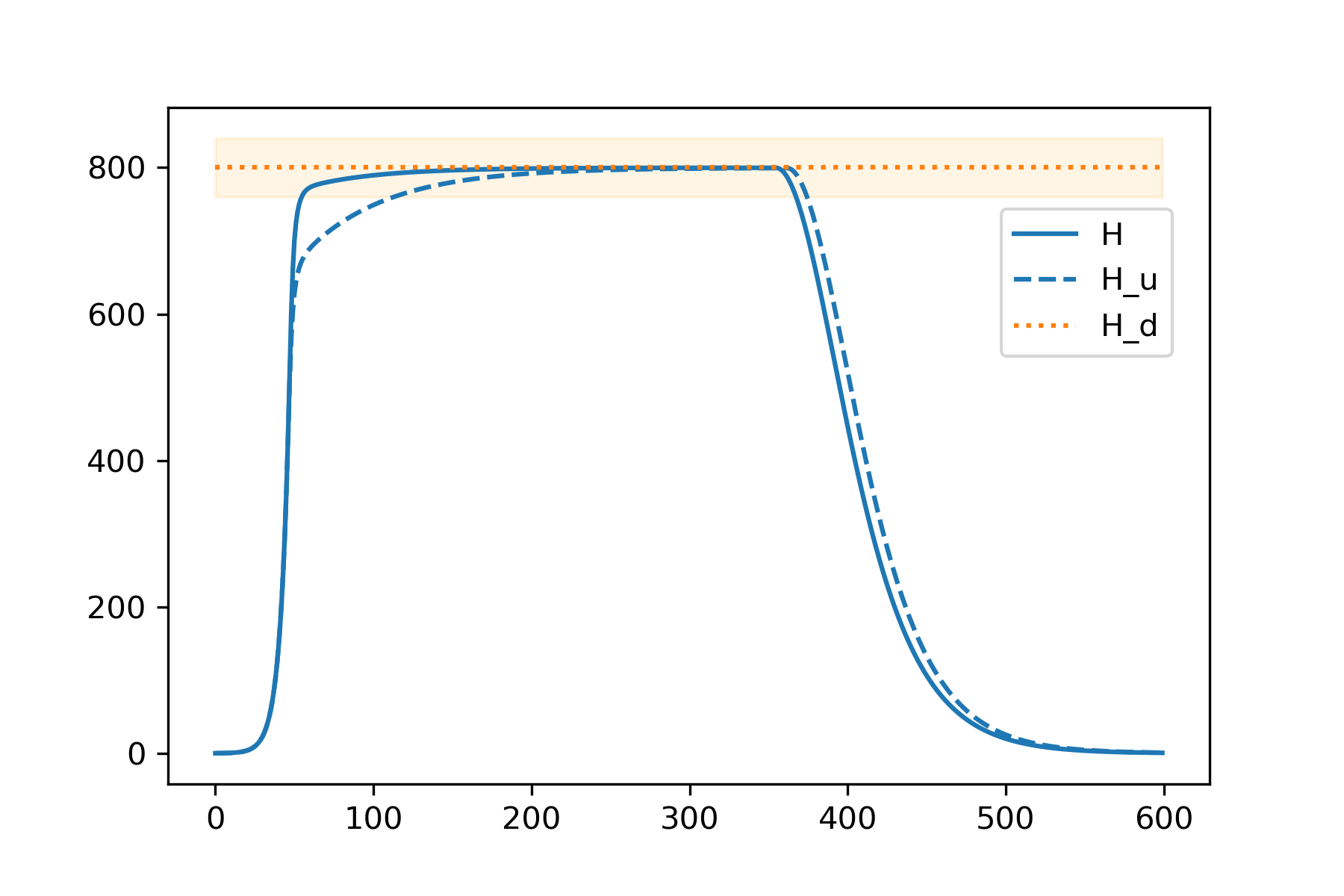}}
\:
\subfigure[$\{SD,SD\_u\}(\text{\%}) \times t(\text{days})$]{\includegraphics[width=0.48\linewidth, height=!]{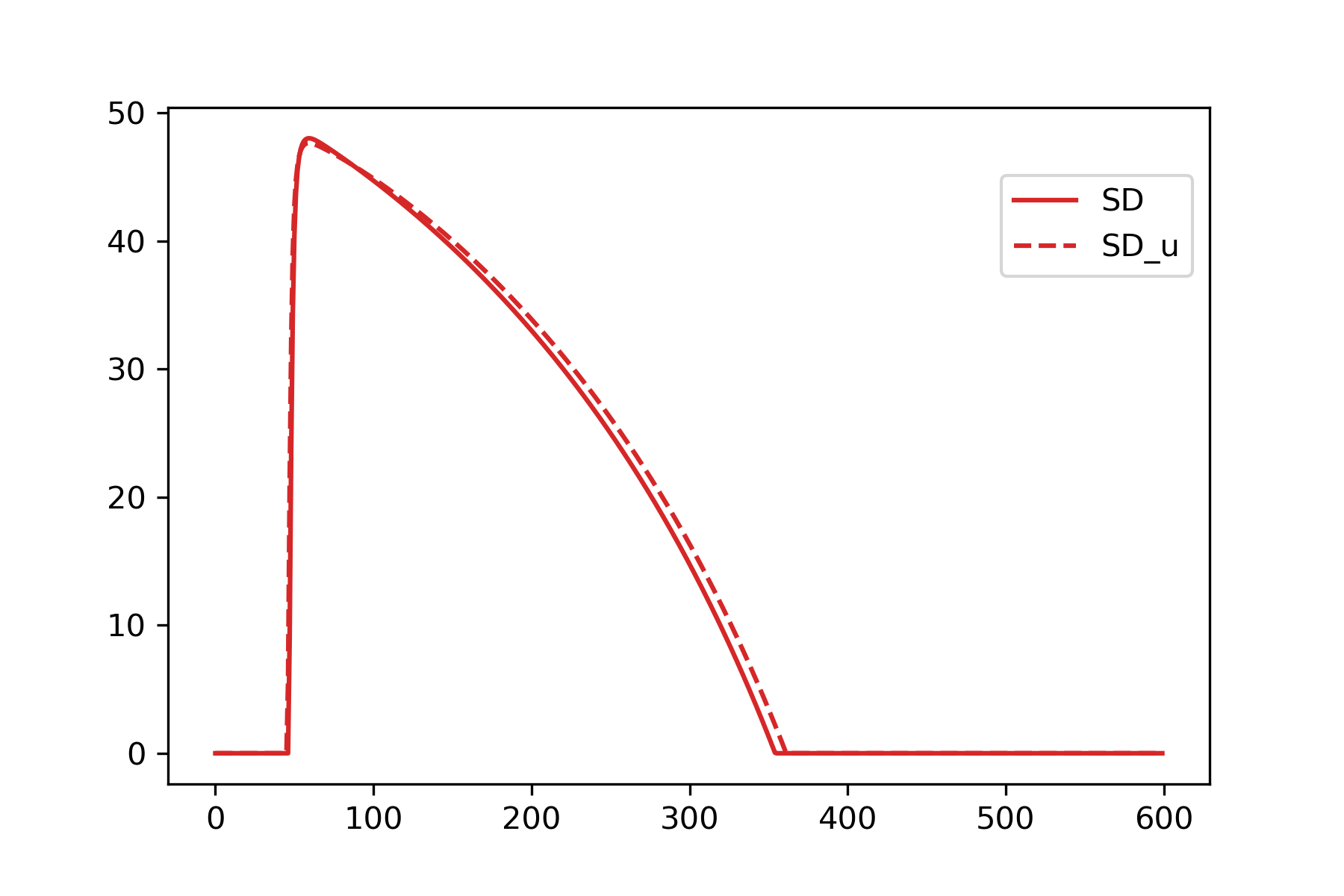}}
\caption{Simulation of standard SIR model with the proposed controller, where blue line is the number of the hospitalized individuals, dashed blue line is the number of the hospitalized individuals when there are uncertainties in the parameters, orange dotted line is the desired number of hospitalized individuals, the light orange fill indicates the range of $\pm 5\%$ of the $H\_d$ value, red line is the percent of the social distancing required and dashed red line is the percent of the social distancing required when there are uncertainties in the parameters}\label{fg:502}
\end{figure}

The simulation of Figure \ref{fg:503} shows the behavior of the infected individuals that will need health care when the proposed controller using the estimator is applied to adjust the social distancing level. In this simulation, it is possible to see that behavior of the number of hospitalized individuals, when does not have uncertainties, is similar to behavior of the case using the measurements. However, when there are uncertainties in the parameter $\beta_0$, the error $e(t)$ will keeps inside the range of the desired number of hospitalized individuals. This error occurs because the estimator need to know $\beta_0$ value. The uncertainties considered was +5\% and -20\% in the values of $\beta_{0}$ and $\gamma$, respectively.    
\begin{figure}[ht]
\center
\subfigure[$\{H^{*},H^{*}\_u,H\_d\}(\text{individuals})\times t(\text{days})$]{\includegraphics[width=0.48\linewidth, height=!]{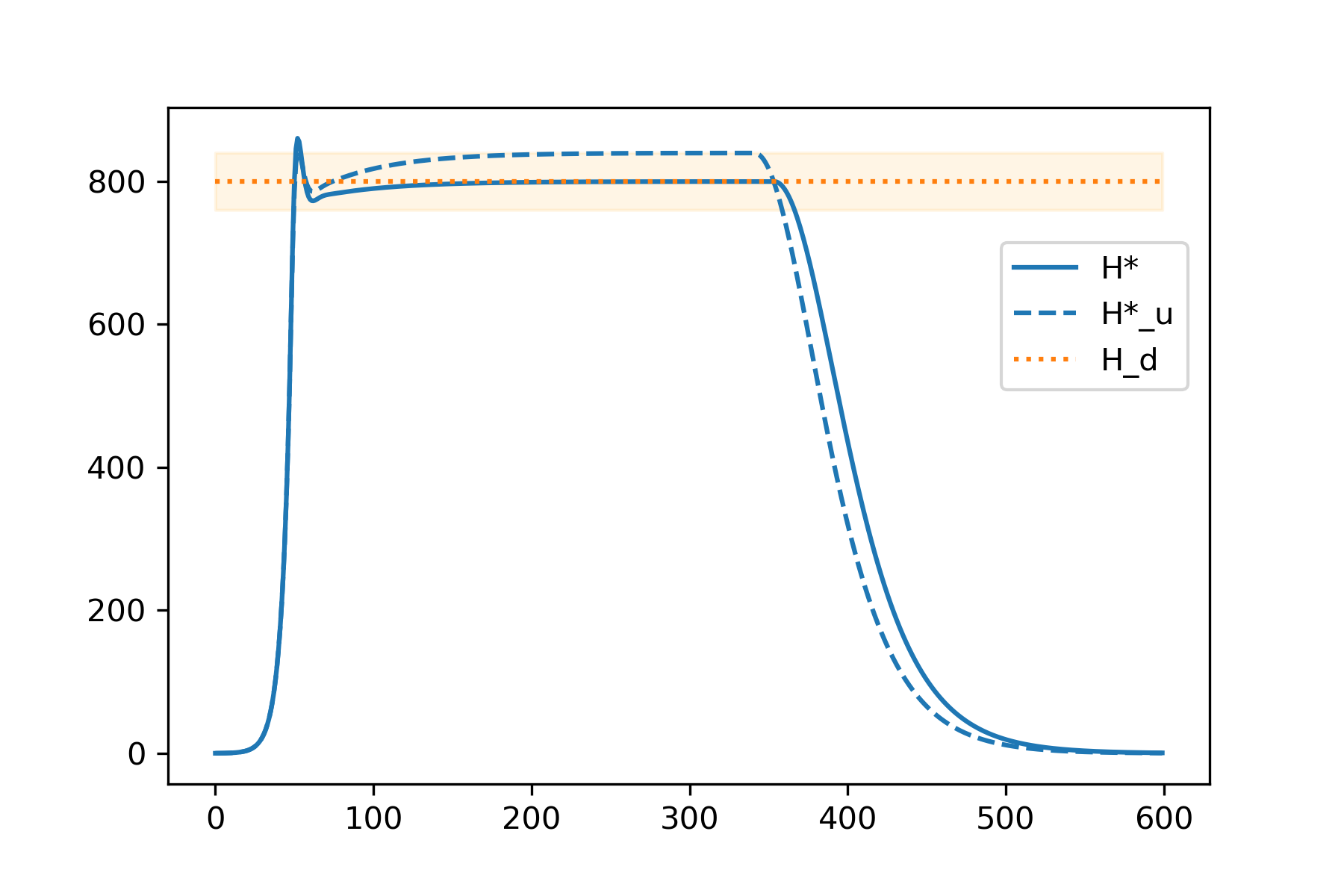}}
\:
\subfigure[$\{SD^{*},SD^{*}\_u\}(\text{\%})\times t(\text{days})$]{\includegraphics[width=0.48\linewidth, height=!]{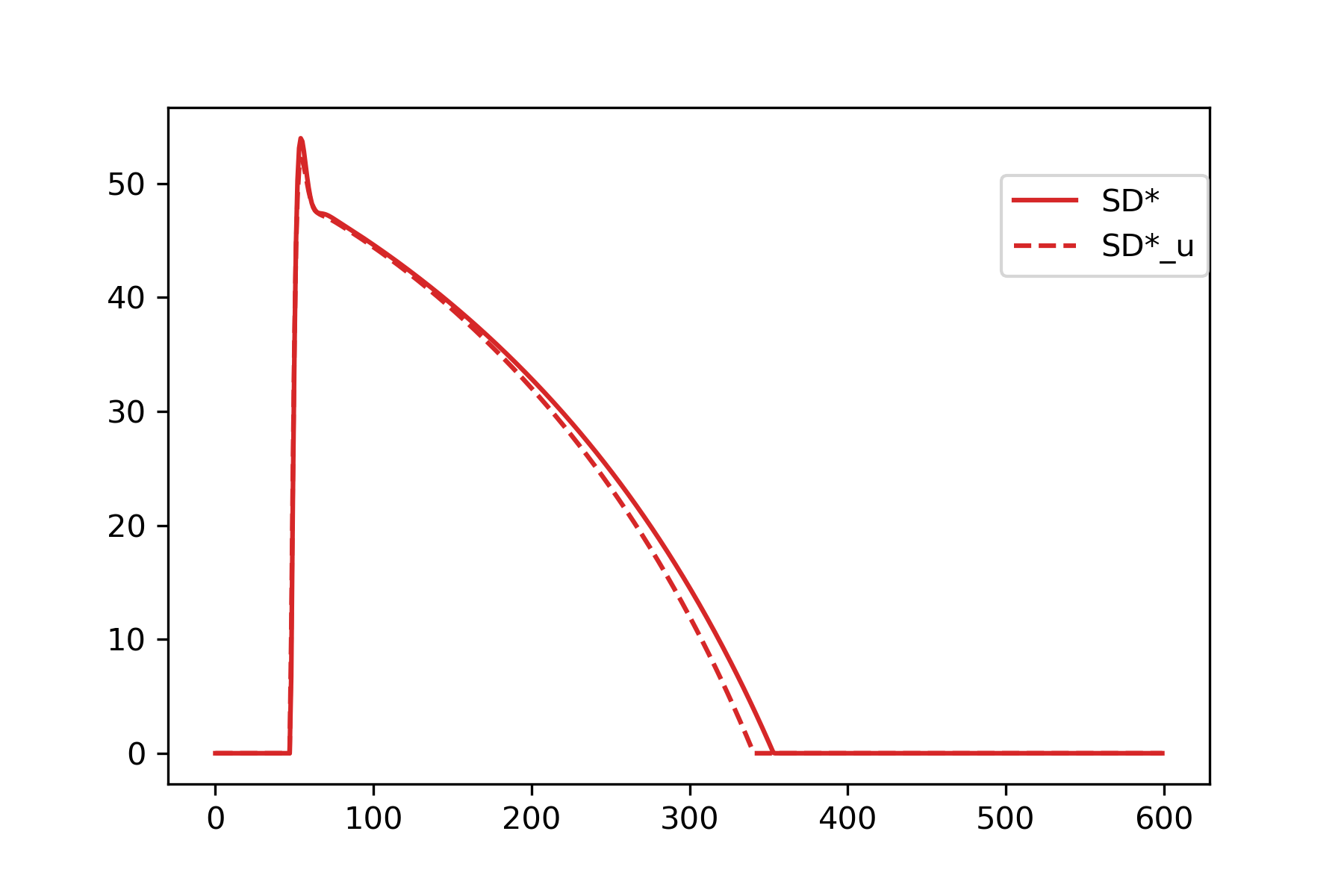}}
\caption{Simulation of standard SIR model with the proposed controller using the estimator, where blue line is the number of the hospitalized individuals, dashed blue line is the number of the hospitalized individuals when there are uncertainties in the parameters, orange dotted line is the desired number of hospitalized individuals, the light orange fill indicates the range of $\pm 5\%$ of the $H\_d$ value, red line is the percent of the social distancing required and dashed red line is the percent of the social distancing required when there are uncertainties in the parameters}\label{fg:503}
\end{figure}

\subsection{Discussion}

The proposed controller calculates the social distancing level to keep the COVID-19 outbreak controlled. The social distancing level is adjusted to guarantee the fastest way to finish the outbreak with the number of hospitalized individuals below the desired value. This technique can reduce the problems of social distancing and keeps the health care system working.

The value of the gains of the proposed controller is chosen by the designer and has the following logic: $\psi_1$ is related to how fast the error ($e(t)$) goes to nearby of zero and $\psi_2$ is related to how softly the error ($e(t)$) converges to zero. 

To get a better result when the estimate of the $i(t)$ is used, it will be necessary to improve the quality of the estimator or a better value of the $\beta_0$.

The necessary time to finish the social distancing depends on the capacity of the health care system to attend the infected individuals.  

\section{Conclusions}

COVID-19, a contact-transmissible infectious disease, is thought to spread through a population via direct contact between individuals \cite{Prem_2020}. Models that assess the effectiveness of physical distancing interventions, such as school closure, need to account for social structures and heterogeneities in the mixing of individuals \cite{Prem_2020,Hilton_2019}.

This work proposed a simple control law to keep the infected individuals during the COVID-19 outbreak below of the desired number. The proposed control law keeps the value of infected individuals controlled only adjusting the social distancing level. The analysis of the stability of the proposed controller was done. 

In this work, simulation results are presented to show as the control law works. In all simulations, the proposed control law reaches its objectives. It was done simulations of the proposed controller to the standard SIR model. The proposed controller was robust to uncertainties in the basic reproductive number $R_{0}$ and disease transmission rate $\gamma$. A simulation considering that only the accumulative number of infected individuals is measured was done and the proposed controller kept the number of infected individuals below the upper limit of the desired number of infected individuals.    

A version of the proposed controller to daily update was developed. This version is too simple that can be calculated in a spreadsheet. This can be considered a good contribution because can easily spread a way to control the COVID-19 outbreak or help in the analysis. 

To conclude, this work proposes solutions for one of the main research challenges of the analysis and control of epidemics. It proposed a controller that considers uncertainties in the parameters.

To the future, a group-structure SIR model will be proposed and the proposed controller will be applied to it. This will be done to develop a decentralized solution and to proposes a more general epidemic model.

\bibliographystyle{unsrt}
\bibliography{references}  

\end{document}